\documentclass[a4paper,UKenglish]{lipics-v2018}
\usepackage{microtype}%if unwanted, comment out or use option "draft"

\usepackage[numbers,sort]{natbib}
\usepackage{todonotes}
\usepackage{xspace,xargs}
\usepackage[ruled,vlined,lined,boxed,commentsnumbered]{algorithm2e}

\newcommand{\repo}{\url{https://github.com/nicolaprezza/bwt2lcp}\xspace}
\newcommand{\eGap}{\texttt{eGap}\xspace}
\title{Space-Efficient Computation of the LCP Array from the Burrows-Wheeler Transform}

\titlerunning{Inducing the LCP from the BWT}

\author{Nicola Prezza}{Department of Computer Science, University of Pisa, Italy}{nicola.prezza@di.unipi.it}{0000-0003-3553-4953
}{}
\author{Giovanna Rosone}
{Department of Computer Science, University of Pisa, Italy}{giovanna.rosone@unipi.it}{0000-0001-5075-1214}{}%mandatory, please use full name; only 1 author per \author macro; first two parameters are mandatory, other parameters can be empty.

\authorrunning{N. Prezza and G. Rosone}%mandatory. First: Use abbreviated first/middle names. Second (only in severe cases): Use first author plus 'et al.'

\Copyright{Nicola Prezza and Giovanna Rosone}%mandatory, please use full first names. LIPIcs license is "CC-BY";  http://creativecommons.org/licenses/by/3.0/

\subjclass{Theory of computation $\rightarrow$ Design and analysis of algorithms $\rightarrow$ Data structures design and analysis}% mandatory: Please choose ACM 2012 classifications from https://www.acm.org/publications/class-2012 or https://dl.acm.org/ccs/ccs_flat.cfm . E.g., cite as "General and reference $\rightarrow$ General literature" or \ccsdesc[100]{General and reference~General literature}. 

%mandatory

\category{}%optional, e.g. invited paper

\relatedversion{}%optional, e.g. full version hosted on arXiv, HAL, or other respository/website

\supplement{}%optional, e.g. related research data, source code, ... hosted on a repository like zenodo, figshare, GitHub, ...

\funding{ GR is partially, and NP is totally, supported by the project MIUR-SIR CMACBioSeq (``Combinatorial methods for analysis and compression of biological sequences'') grant n.~RBSI146R5L.}%optional, to capture a funding statement, which applies to all authors. Please enter author specific funding statements as fifth argument of the \author macro.

\EventEditors{John Q. Open and Joan R. Access}
\EventNoEds{2}
\EventLongTitle{42nd Conference on Very Important Topics (CVIT 2016)}
\EventShortTitle{CVIT 2016}
\EventAcronym{CVIT}
\EventYear{2016}
\EventDate{December 24--27, 2016}
\EventLocation{Little Whinging, United Kingdom}
\EventLogo{}
\SeriesVolume{42}
\ArticleNo{23}
\nolinenumbers %uncomment to disable line numbering
\def\mergeBwtLCP{\texttt{merge}\xspace}
\def\induceLCP{\texttt{bwt2lcp}\xspace}

\begin{document}
	
	\maketitle
	
	\begin{abstract}
		
		We show that the Longest Common Prefix Array of a text collection of total size $n$ on alphabet $[1,\sigma]$ can be computed from the Burrows-Wheeler transformed collection in $O(n\log\sigma)$ time using $o(n\log\sigma)$ bits of working space on top of the input and output. 
		Our result improves (on small alphabets) and generalizes (to string collections) the previous solution from Beller et al., which required $O(n)$ bits of extra working space. 
		We also show how to merge the BWTs of two collections of total size $n$ within the same time and space bounds. The procedure at the core of our algorithms can be used to enumerate suffix tree intervals in succinct space from the BWT, which is of independent interest.
		%and (optionally) induce the LCP array of their union. 
		%Also this result beats the state of the art (from Belazzougui et al.), which merges BWTs using $n\log\sigma + O(n)$ bits of working space. 
		An engineered implementation of our first algorithm on DNA alphabet induces the LCP of a large (16 GiB) collection of short (100 bases) reads at a rate of $2.92$ megabases per second using in total $1.5$ Bytes per base in RAM. 
		Our second algorithm merges the BWTs of two short-reads collections of 8 GiB each at a rate of $1.7$ megabases per second and uses $0.625$ Bytes per base in RAM. An extension of this algorithm that computes also the LCP array of the merged collection processes the data at a rate of $1.48$ megabases per second and uses $1.625$ Bytes per base in RAM. 
	\end{abstract}
	
	\section{Introduction}
	The increasingly-growing production of huge datasets composed of short strings---especially in domains such as biology, where new generation sequencing technologies can nowadays generate Gigabytes of data in few hours---is lately generating much interest towards fast and space-efficient algorithms able to index this data. The Burrows-Wheeler Transform~\cite{burrows1994block} and its extension to sets of strings~\cite{MantaciRRS07,BauerCoxRosoneTCS2013} is becoming the gold-standard in the field: even when not compressed, its size is one order of magnitude smaller than classic suffix arrays (while preserving many of their indexing capabilities). 
	The functionalities of this transformation can be extended by computing additional structures such as the LCP array \cite{CGRS_JDA_2016}; see, e.g. \cite{prezza2018detecting} for a bioinformatics application where this component is needed. To date, several practical algorithms have been developed to solve the task of merging or building \emph{de novo} such components~\cite{BauerCoxRosoneTCS2013,CGRS_JDA_2016,holt2014constructing,holt2014merging,Bonizzoni2018,egidiLouzaManziniTellesWABI2018}, but little work has been devoted to the task of computing the LCP array from the BWT of string collections in little space (internal and external working space). The only existing work we are aware of in this direction is from Beller et al.~\cite{beller2013computing}, who show how to build the LCP array from the BWT of a single text in $O(n\log\sigma)$ time and $O(n)$ bits of working space on top of the input and output. In this paper, we combine their algorithm with a recent suffix-tree enumeration procedure of Belazzougui~\cite{belazzougui2014linear} and reduce this working space to $o(n\log\sigma)$ while also generalizing the algorithm to string collections. As a by-product, we show an algorithm able to merge the BWTs of two string collections using just $o(n\log\sigma)$ bits of working space. An efficient implementation of our algorithms on DNA alphabet uses (in RAM) as few as $n$ bits on top of a succinct representation of the input/output, and can process data as fast as $2.92$ megabases per second.
	
	%\eGap~\cite{egidi2017lightweight} is a semi-external algorithm designed for the same task of merging BWTs of collections while inducing the LCP of their union. Letting $\mathtt{avelcp}$ being the average of the LCP array values, \eGap runs in $O(n\cdot \mathtt{avelcp})$ time and uses, in the worst case, $\Theta(n\log n)$ bits of total working space in addition to the input and output (this space can be lowered by slowing down the algorithm). Importantly, \eGap is a semi-external algorithm (while our algorithm uses only internal memory) and therefore its internal memory usage could be lower than ours. 

\section{Our Contributions}

	Our work builds upon the following two results from Belazzougui\footnote{While the original theorem~\cite[Sec. 5.1]{belazzougui2014linear} is general and uses the underlying rank data structure as a black box, in our case we strive for succinct space (not compact as in~\cite{belazzougui2014linear}) and stick to wavelet trees. All details on how to achieve the claimed running time and space are described in Section \ref{sec:belazzougui}.}\cite{belazzougui2014linear} and Beller et al.~\cite{beller2013computing}. 
	For space reasons, the notation used in the paper is reported in Appendix \ref{sec:notation}.	
	%The key insights of these results are described more in detail in the next sections. 
	
	\begin{theorem}[Belazzougui \cite{belazzougui2014linear}]\label{th:Belazzougui}
		Given the Burrows-Wheeler Transform of a text $T\in[1,\sigma]^n$ represented with a wavelet tree, we can solve the following problem in $O(n\log\sigma)$ time using $O(\sigma^2\log^2n)$ bits of working space on top of the BWT. Enumerate the following information for each distinct right-maximal substring $W$ of $T$:
		\begin{itemize}
			\item $|W|$.
			\item $range(Wc_i)$ for all $c_1 < \dots < c_k$ such that $Wc_i$ occurs in $T$.
		\end{itemize}
	\end{theorem}
	
	\begin{theorem}[Beller et al.\cite{beller2013computing}]\label{th:Beller}
		Given the Burrows-Wheeler Transform of a text $T$ represented with a wavelet tree, we can compute the LCP array of $T$ in $O(n\log\sigma)$ time using $4n$ bits of working space on top of the BWT and the LCP.
	\end{theorem}

	Theorem \ref{th:Beller} represents the state of the art for computing the LCP array from the BWT. 
	%Note that Theorem \ref{th:Belazzougui} does not specify the order in which right-maximal strings are enumerated. 
	Our first observation is that Theorem \ref{th:Belazzougui} can be directly used to induce the LCP array of $T$ using just $O(\sigma^2\log^2n)$ bits of working space on top of the input and output (proof in Section \ref{sec:algo1}). We combine this result with Theorem \ref{th:Beller} and obtain our first theorem:
	
	\begin{theorem}\label{thm:LCP}
		Given the Burrows-Wheeler Transform of a text $T\in[1,\sigma]^n$, we can compute the LCP array of $T$ in $O(n\log\sigma)$ time using $o(n\log\sigma)$ bits of working space on top of the BWT and the LCP.
	\end{theorem}
	\begin{proof}
		
		First, we replace $T$ by its wavelet matrix~\cite{claude2015wavelet} --- of size $n\log\sigma + o(n\log\sigma)$ bits --- in $O(n\log\sigma)$ time using just $n$ bits of additional working space as shown in~\cite{claude2015wavelet}. Wavelet matrices support the same set of operations of wavelet trees in the same running times (indeed, they can be considered as a wavelet tree representation). We re-use the space of the LCP array to accommodate the extra $n$ bits required for building the wavelet matrix, so the overall working space does not exceed $o(n\log\sigma)$ bits on top of the BWT and LCP. 
		In the rest of the paper we will simply assume that the input is represented by a wavelet tree.
		
		At this point, if $\sigma < \sqrt{n}/\log^2n$ then $\sigma^2\log^2n = o(n)$ and our extension of Theorem \ref{th:Belazzougui} gives us $o(n\log\sigma)$ additional working space. If $\sigma \geq \sqrt{n}/\log^2n$ then $\log\sigma = \Theta(\log n)$ and we can use Theorem \ref{th:Beller}, which yields extra working space $O(n) = o(n\log n) = o(n\log\sigma)$. 
	\end{proof}
	
	We proceed by extending Theorem \ref{th:Belazzougui} to enumerate also the intervals corresponding to leaves of the generalized suffix tree of a text collection (Theorem \ref{th:Belazzougui} enumerates internal nodes).
	We show that this simple modification, combined again with the strategy of Theorem \ref{th:Beller} (generalized to text collections), can be used to extend Theorem \ref{thm:LCP} to text collections: 
	
	\begin{theorem}\label{thm:LCP collection}
		Given the Burrows-Wheeler Transform of a collection $\mathcal C = \{T_1, \dots, T_m\}$ of total length $n$ on alphabet $[1,\sigma]$, we can compute the LCP array of $\mathcal C$ in $O(n\log\sigma)$ time using $o(n\log\sigma)$ bits of working space on top of the BWT and the LCP.
	\end{theorem}
	
	In~\cite{belazzougui2014linear,belazzougui2016linear}, Belazzougui et al. show that Theorem \ref{th:Belazzougui} can be adapted to merge the BWTs of two texts $T_1, T_2$ and obtain the BWT of the collection $\{T_1, T_2\}$ in $O(nk)$ time and $n\log\sigma(1+1/k) + 11n + o(n)$ bits of working space for any $k \geq 1$~\cite[Thm. 7]{belazzougui2016linear}. We show that our strategy enables a more space-efficient algorithm for the task of merging BWTs of collections.
	%and, at the same time, computing the LCP array of the merged collection. 
	The following theorem merges two BWTs by computing the binary Document Array (DA) of their union, i.e. a bitvector telling whether the $i$-th suffix comes from the first or second collection. After that, the merged BWT can be streamed to external memory (the DA tells how to interleave characters from the input BWTs) and does not take additional space in internal memory. Similarly to what we did in the proof of Theorem \ref{thm:LCP}, this time we re-use the space of the Document Array to accommodate the extra $n$ bits needed to replace the BWTs of the two collections with their wavelet matrices.
	
	\begin{theorem}\label{th:merge}
		Given the Burrows-Wheeler Transforms of two collections $\mathcal C_1$ and $\mathcal C_2$ of total length $n$ on alphabet $[1,\sigma]$, we can compute the Document Array 
		%and (optionally) the LCP array 
		of $\mathcal C_1 \cup \mathcal C_2$ in $O(n\log\sigma)$ time using $o(n\log\sigma)$ bits of working space on top of the input BWTs and the output DA.
	\end{theorem}
	
	When $k=\log\sigma$, the running time of~\cite[Thm. 7]{belazzougui2016linear} is the same as our Theorem \ref{th:merge} but the working space is higher: $n\log\sigma + O(n)$ bits.
	%Our working space is also lower than the worst-case of $O(n\log n)$ bits of \eGap by Egidi and Manzini~\cite{egidi2017lightweight} (which however can be reduced at expenses of running time).	
	%An extension of Theorem \ref{th:merge} to $m\geq 2$ collections is immediate but omitted here for space constraints. 
	We also briefly discuss how to extend Theorem \ref{th:merge} to build the LCP array of the merged collection.
	In Section \ref{sec:experiments} we present an  implementation of our algorithms and an experimental comparison  with \eGap~\cite{egidi2017lightweight}, the state-of-the-art tool designed for the same task of merging BWTs while inducing the LCP of their union.

	\section{Belazzougui's Enumeration Algorithm}\label{sec:belazzougui}
	In~\cite{belazzougui2014linear}, Belazzougui showed that a BWT with \emph{rank} and \emph{range distinct} functionality (see Appendix \ref{sec:notation}) is sufficient to enumerate in small space a rich representation of the internal nodes of the suffix tree of a text $T$.
	In this section we describe his algorithm. 
	
	Remember that explicit suffix tree nodes correspond to right-maximal text substrings. By definition, for any right-maximal substring $W$ there exist at least two distinct characters $c_1, \dots, c_k$ such that $Wc_i$ is a substring of $T$, for $i=1, \dots, k$. 
	The first idea is to represent any text substring $W$ (not necessarily right-maximal) as follows. Let $\mathtt{chars_W[1] < \dots < chars_W[k_W]}$ be the (lexicographically-sorted) character array such that $W\cdot \mathtt{chars_W[i]}$ is a substring of $T$ for all\footnote{We require $\mathtt{chars_W}$ to be also complete: if $Wc$ is a substring of $T$, then $c\in \mathtt{chars_W}$} $i=1,\dots, k_W$, where $k_W$ is the number of right-extensions of $W$.
	Let moreover $\mathtt{first_W[1..k_W+1]}$ be the array such that $\mathtt{first_W[i]}$ is the starting position of (the range of) $W\cdot \mathtt{chars_W[i]}$ in the suffix array of $T$ for $i=1,\dots, k_W$, and $\mathtt{first_W[k_W+1]}$ is the end position of $W$ in the suffix array of $T$. The representation for $W$ is (differently from~\cite{belazzougui2014linear}, we omit $\mathtt{chars_W}$ from the representation and we add $|W|$; these modifications will turn useful later):
	$$
	\mathtt{repr(W) = \langle \mathtt{first_W},\ |W| \rangle}
	$$
	Note that, if $W$ is not right-maximal and is not a text suffix, then $W$ is followed by $k_W=1$ distinct characters in $T$ and the above representation is still well-defined.
	When $W$ is right-maximal, we will also say that $\mathtt{repr(W)}$ is the representation of a suffix tree explicit node (i.e. the node reached by following the path labeled $W$ from the root). 
	At this point, the enumeration algorithm works by visiting the Weiner Link tree of $T$ starting from the root's representation $\mathtt{repr(\epsilon) = \langle \mathtt{first_\epsilon},\ 0 \rangle}$, where $\mathtt{first_\epsilon} = \langle C[c_1], \dots, C[c_\sigma], n \rangle$ (see Appendix \ref{sec:notation} for a definition of the $C$-array) and $c_1, \dots, c_\sigma$ 
	are all (and only) the sorted alphabet's characters.
	The visit uses a stack storing representations of suffix tree nodes, initialized with  $\mathtt{repr(\epsilon)}$. At each iteration, we pop the head $\mathtt{repr(W)}$ from the stack and we push $\mathtt{repr(cW)}$ such that $cW$ is right-maximal in $T$. 
	If nodes are pushed on the stack in decreasing order of interval length, then the stack's size never exceeds $O(\sigma\log n)$.
	For completeness, in Appendix \ref{app:belazzougui} we describe in detail how Weiner links are computed, and show that with this strategy we visit all suffix tree nodes in $O(n\log\sigma)$ time using overall $O(\sigma^2\log^2 n)$ bits of additional space (for the stack). 
	In Section \ref{sec:algo1} we show that this enumeration algorithm can be used to compute the LCP array from the BWT of a collection. 
	
	\section{Beller et al.'s Algorithm}\label{sec:beller}
	
	Also Beller et al.'s algorithm~\cite{beller2013computing} works by enumerating a (linear) subset of the BWT intervals. LCP values are induced from a particular visit of those intervals. Belazzougui's and Beller et al.'s algorithms have, however, two key differences which make the former more space-efficient on small alphabets, while the latter more space-efficient on large alphabets: 	(i) Beller et al. use a queue (FIFO) instead of a stack (LIFO), and (ii) they represent $W$-intervals with just the pair of coordinates $\mathtt{range(W)}$ and the value $|W|$.
	In short, while Beller et al.'s queue might grow up to size $\Theta(n)$, the use of intervals (instead of the more complex representation used by Belazzougui) makes it possible to represent it using $O(1)$ bitvectors of length $n$. On the other hand, the size of Belazzougui's stack can be upper-bounded by $O(\sigma\log n)$, but its elements take more space to be represented. 
	
	Beller et al.'s algorithm starts by initializing all LCP entries to $\bot$ (an undefined value), and by inserting in the queue the triple $\langle1,n,0 \rangle$, where the first two components are the BWT interval of $\epsilon$ (the empty string) and the third component is its length. From this point, the algorithm keeps performing the following operations until the queue is empty. We remove the first (i.e. the oldest) element $\langle L,R,\ell\rangle$ from the queue, which (by induction) is the interval and length of some string $W$: $\mathtt{range(W)}= \langle L,R \rangle$ and $|W|=\ell$. 
	Using  operation $\mathtt{getIntervals(L,R,BWT)}$~\cite{beller2013computing} (see Appendix \ref{sec:notation})
	we left-extend the BWT interval $\langle L,R\rangle$ with the characters $c_1, \dots, c_k$ in $\mathtt{rangeDistinct(L,R)}$, obtaining the triples $\langle L_1, R_1, \ell+1 \rangle, \dots, \langle L_k, R_k, \ell+1 \rangle$ corresponding to the strings $c_1W, \dots, c_kW$. For each such triple $\langle L_i, R_i, \ell+1\rangle$, if $R_i\neq n$ and $LCP[R_i+1] = \bot$ then we set $LCP[R_i+1] \leftarrow \ell$ and push $\langle L_i, R_i, \ell+1\rangle$ on the queue. Importantly, note that we can push the intervals returned by $\mathtt{getIntervals(L,R,BWT)}$ in the queue in any order; as discussed in Appendix \ref{sec:notation}, this step can be implemented with just $O(\log n)$ bits of space overhead with a DFS-visit of the wavelet tree's sub-tree induced by $BWT[L,R]$ (i.e. the intervals are not stored temporarily anywhere: they are pushed as soon as they are generated).
	To limit space usage, Beller et al. use two different queue representations. As long as there are $O(n/\log n)$ elements in the queue, they use a simple vector. When there are more intervals, they switch to a representation based on four bitvectors of length $n$ that still guarantees constant amortized operations. All details are described in Appendix \ref{app:beller}. Beller et al.~\cite{beller2013computing} show that the above algorithm correctly computes the LCP array of a text. In the next section we generalize the algorithm to text collections. 
	
	%A proof of correctness and completeness of the above Algorithm can be found in Section \ref{sec:algo1}, where we treat the more general case of text collections.
	
	\section{Our Algorithms}
	
	%In this section we describe our space-efficient algorithms. 
	%The former (Section \ref{sec:algo1}) computes the LCP array from the wavelet tree of the Burrows-Wheeler transformed collection (which, as observed in the proof of Theorem \ref{thm:LCP}, we can compute by replacing the BWT). The latter (Section \ref{sec:algo2}) merges the Burrows-Wheeler transforms of two collections and can be extended to also compute the LCP array of their union.  
	We describe our algorithms directly on string collections. This will include, as a particular case, inputs formed by a single text. 
	Procedure \texttt{BGOS(BWT,LCP)} in Line \ref{beller et al.} of Algorithm \ref{alg:fill nodes} is a call to Beller et al.'s algorithm, modified as follows. First, we set $\mathtt{LCP[C[c]]}\leftarrow 0$ for all $c\in\Sigma$. Then, we push in the queue $\langle \mathtt{range(c), 1} \rangle$ for all $c\in\Sigma$ and start the main algorithm. Note moreover that (see Appendix \ref{sec:notation}) from now on we never left-extend ranges with $\#$.
	
	%In our algorithms (including the call \texttt{BGOS(BWT,LCP)} to Beller et al.'s algorithm in Line \ref{beller et al.} of Algorithm \ref{alg:fill nodes}), we assume that \texttt{BWT.range\_distinct(L,R)} returns all distinct characters \emph{different than} $\#$ (the terminator) in $BWT[L,R]$. Similarly, $\mathtt{getIntervals(L,R, BWT)}$ returns on-the-fly (i.e. without storing them in a temporary vector) all SA-intervals of left-extensions $cW$, with $c\neq \#$ of the string $W$ whose interval is $\langle L,R\rangle$. See~\cite{beller2013computing} and Appendix \ref{sec:notation} for all details describing how these functions are implemented on wavelet trees. 
	%This is crucial in order for our extension to string collections to work properly. 
	
	\subsection{Computing the LCP From the BWT}\label{sec:algo1}
	
	Let $\mathcal C$ be a text collection where each string is ended by a terminator $\#$ (common to all strings). 
	Consider now the LCP and GSA (generalized Suffix Array) arrays of $\mathcal C$. We divide LCP values in two types. 
	Let $GSA[i] = \langle j,k \rangle$, with $i>1$, indicate that the $i$-th suffix in the lexicographic ordering of all suffixes of strings in $\mathcal C$ is $\mathcal C[j][k..]$. A LCP value $\mathtt{LCP[i]}$ is of \emph{node type} when the $i$-th and $(i-1)$-th suffixes are distinct: $\mathcal C[j][k..] \neq \mathcal C[j'][k'..]$, where $GSA[i] = \langle j,k \rangle$ and $GSA[i-1] = \langle j',k' \rangle$. Those two suffixes differ before the terminator is reached in both suffixes (it might be reached in one of the two suffixes, however); we use the name \emph{node-type} because $i-1$ and $i$ are the last and first suffix array positions of the ranges of two adjacent children of some suffix tree node, respectively (i.e. the node corresponding to string $\mathcal C[j][k..k+LCP[i]-1]$).
	Note that it might be that one of the two suffixes, $\mathcal C[j][k..]$ or $\mathcal C[j'][k'..]$, is the empty string (followed by the terminator) $\#$.
	Similarly, a \emph{leaf-type} LCP value $\mathtt{LCP[i]}$ is such that the $i$-th and $(i-1)$-th suffixes are equal: $\mathcal C[j][k..] = \mathcal C[j'][k'..]$. We use the name \emph{leaf-type} because, in this case, it must be the case that $i \in [L+1,R]$, where $\langle L,R \rangle$ is the suffix array range of some suffix tree leaf (it might be that $R>L$ since there might be repeated suffixes in the collection). Note that, in this case, $\mathcal C[j][k..] = \mathcal C[j'][k'..]$ could coincide with $\#$. Entry $LCP[0]$ escapes the above classification, so we will set it separately.
	
	Our idea is to compute first node-type and then leaf-type LCP values. We argue that Beller et al.'s algorithm already computes the former kind of LCP values. When this algorithm uses too much space (i.e. on small alphabets), we show that Belazzougui's enumeration strategy can be adapted to reach the same goal: by the very definition of node-type LCP values, they lie between children of some suffix tree node $x$, and their value corresponds to the string depth of $x$. This strategy is described in Algorithm \ref{alg:fill nodes}. Function $\mathtt{BWT.Weiner(x)}$ in Line \ref{range distinct2} takes as input the representation of a suffix tree node $x$ and returns all explicit nodes reached by following Weiner links form $x$ (an implementation of this function is described in Appendix \ref{app:belazzougui}).
	Leaf-type LCP values, on the other hand, can easily be computed by enumerating intervals corresponding to suffix tree leaves. To reach this goal, it is sufficient to enumerate ranges of suffix tree leaves starting from $\mathtt{range(\#)}$ and recursively left-extending with backward search with characters different than $\#$ whenever possible. For each range $\langle L,R \rangle$ obtained in this way, we set each entry $LCP[L+1,R]$ to the string depth (terminator excluded) of the corresponding leaf. 
	This strategy is described in Algorithm \ref{alg:fill leaves}.
	In order to limit space usage, we use again a stack or a queue to store leaves and their string depth (note that each leaf takes $O(\log n)$ bits to be represented): we use a queue when $\sigma > n/\log^3n$, and a stack otherwise. 	
	This guarantees that the bit-size of the queue/stack never exceeds $o(n\log\sigma)$ bits.
	The queue is the same used by Beller et al.\cite{beller2013computing} and described in Appendix \ref{app:beller}.
	%Again, we note that (as discussed in Appendix \ref{sec:notation}), procedure $\mathtt{getIntervals(L,R,BWT)}$ in Line \ref{push7} returns intervals on-the-fly through a visit of the wavelet tree's subtree induced by $BWT[L,R]$; in particular, the intervals are pushed in the queue as soon as they are visited and never stored anywhere (except in the queue). 
	Note that in Lines \ref{getIntervals1}-\ref{push2} we can afford storing temporarily the $k$ resulting intervals since, in this case, the alphabet's size is small enough. 
	To sum up, our full procedure works as follows: 
	\begin{enumerate}
		\item We initialize an empty array $\mathtt{LCP[1..n]}$.
		\item We fill node-type entries using procedure \texttt{Node-Type}$(\mathtt{BWT, LCP})$ described in Algorithm \ref{alg:fill nodes}. 
		\item We fill leaf-type entries using procedure \texttt{Leaf-Type}$(\mathtt{BWT, LCP})$ described in Algorithm \ref{alg:fill leaves}.
	\end{enumerate}

	\begin{algorithm}[th!]
		\caption{\texttt{Node-Type(BWT,LCP)}}
		\label{alg:fill nodes}
		
		\SetKwInOut{Input}{input}
		\SetKwInOut{Output}{behavior}
		\SetSideCommentLeft
		\LinesNumbered
		
		\Input{Wavelet tree of the Burrows-Wheeler transformed collection  $\mathtt{BWT}\in [1,\sigma]^n$ and empty array $\mathtt{LCP[1..n]}$.}
		\Output{Fills node-type $\mathtt{LCP}$ values.}
		\BlankLine
		\BlankLine
		
		\eIf{$\sigma > \sqrt n/\log^2n$}{
			
			$\mathtt{BGOS(BWT,LCP)}$\tcc*[r]{Run Beller et al.'s algorithm}\label{beller et al.}
			
		}{
			
			$\mathtt P \leftarrow \texttt{new\_stack()}$\tcc*[r]{Initialize new stack}\label{new stack2}	
			
			\BlankLine
			
			$\mathtt P\mathtt{.push( repr(\epsilon))}$\tcc*[r]{Push representation of $\epsilon$}\label{push3}
			
			\BlankLine
			
			\While{$\mathtt{\mathbf{not}\ P.empty()}$}{\label{while2}
				\BlankLine
				
				$\langle \mathtt{first_W},\ \ell \rangle \leftarrow \mathtt{P.pop()}$\tcc*[r]{Pop highest-priority element}\label{pop2}	
				$t \leftarrow |\mathtt{first_W}|-1$\tcc*[r]{Number of children of ST node}\label{nchild}
				
				\BlankLine
				
				\For{$i = 2, \dots, t$}{
					$\mathtt{LCP}[\mathtt{first_W}[i]] \leftarrow \ell$\tcc*[r]{Set LCP value}\label{LCP in Node}
				}	
				
				\BlankLine
				
				$x_1,\dots, x_k \leftarrow \mathtt{BWT.Weiner(\langle first_W,\ \ell \rangle)}$\tcc*[r]{Follow Weiner Links}\label{range distinct2}
				
				\BlankLine
				
				$x'_1, \dots, x'_k \leftarrow \mathtt{sort}(x_1, \dots, x_k)$\tcc*[r]{Sort by interval length}\label{sort2}
				
				\BlankLine
				
				\For{$i=k\dots 1$}{
					
					$\mathtt{P.push}(x'_i)$\tcc*[r]{Push representations}\label{push4}
					
				}
				
			}
			
		}
		
		\BlankLine

		$\mathtt{LCP[0]} \leftarrow 0$\;\label{LCP[0]}		
		
	\end{algorithm}

	\begin{algorithm}[th!]
		\caption{\texttt{Leaf-Type(BWT,LCP)}}
		\label{alg:fill leaves}
		
		\SetKwInOut{Input}{input}
		\SetKwInOut{Output}{behavior}
		\SetSideCommentLeft
		\LinesNumbered
		
		\Input{Wavelet tree of the Burrows-Wheeler transformed collection $\mathtt{BWT}\in [1,\sigma]^n$ and array $\mathtt{LCP[1..n]}$.}
		\Output{Fills leaf-type $\mathtt{LCP}$ values.}
		\BlankLine	
		
		%$\mathtt{LCP[left(\#),right(\#)]} \leftarrow 0$\;
		
		\BlankLine
		
		\eIf{$\sigma > n/\log^3n$}{
			
			$\mathtt P \leftarrow \mathtt{new\_queue()}$\tcc*[r]{Initialize new queue}\label{new queue1}
			
		}{
			
			$\mathtt P \leftarrow \mathtt{new\_stack()}$\tcc*[r]{Initialize new stack}\label{new stack1}
			
		}
		
		\BlankLine
		
		$\mathtt P\mathtt{.push( BWT.range(\#),0)}$\tcc*[r]{Push range of terminator and LCP value 0}\label{push1}
		
		\BlankLine
		
		\While{$\mathtt{\mathbf{not}\ P.empty()}$}{\label{while1}
			\BlankLine
			
			$\langle \langle L,R \rangle, \ell \rangle \leftarrow \mathtt{P.pop()}$\tcc*[r]{Pop highest-priority element}\label{pop1}	
			
			\BlankLine
			
			\For{$i=L+1\dots R$}{
				$\mathtt{LCP}[i] \leftarrow \ell$\tcc*[r]{Set LCP inside range of ST leaf}\label{LCP in Leaves}
			}	
			
			\BlankLine

			\eIf{$\sigma > n/\log^3n$}{
				
				$\mathtt{P.push(getIntervals(L, R, BWT), \ell+1)}$\tcc*[r]{Pairs $\langle$interval,$\ell+1\rangle$}\label{push7}

			}{

				$\langle L_i, R_i\rangle_{i=1, \dots, k} \leftarrow \mathtt{getIntervals(L, R,BWT)}$\;\label{getIntervals1}
				
				$\langle L'_i, R'_i\rangle_{i=1, \dots, k} \leftarrow \mathtt{sort}(\langle L_i, R_i\rangle_{i=1, \dots, k})$\tcc*[r]{Sort by interval length}\label{sort1}
				
				\BlankLine
				
				\For{$i=k\dots 1$}{
					
					$\mathtt{P.push}(\langle L'_i, R'_i\rangle, \ell+1)$\tcc*[r]{Push in order of decreasing length}\label{push2}
					
				}
	
			}			
		}
	\end{algorithm}

	%\begin{lemma}\label{lem:nav}
	%	Given the Burrows-Wheeler Transform of a collection $\mathcal C = \{T_1, \dots, T_m\}$ of total length $n$ represented with a wavelet tree of size $(1+o(1))\cdot n\log\sigma$ bits, we can solve the following problems in $O(n\log\sigma)$ time using $O(\sigma^2\log^2n)$ bits of working space on top of the BWT: (i) for each distinct right-maximal substring $W$ in the collection, enumerate the information of Theorem \ref{th:Belazzougui}, and (ii) for each distinct suffix $U$ of strings in $\mathcal C$ enumerate:
	%	\begin{itemize}
	%		\item $|U|$.
	%		\item $range(U)$.
	%	\end{itemize}
	%\end{lemma}
	
	%With respect to Theorem \ref{th:Belazzougui}, Lemma \ref{lem:nav} provides an additional procedure for navigating the leaves of the generalized suffix tree (in addition to its internal nodes). 
	
	Theorems \ref{thm:LCP} and \ref{thm:LCP collection} follow from the correctness of our procedure, which for space reasons is reported in Appendix \ref{app:proof} as Lemma \ref{lemma:proof of thm1}.
	As a by-product, in Appendix \ref{sec:ST} we note that Algorithm \ref{alg:fill nodes} can be used to enumerate suffix tree intervals in succinct space from the BWT, which could be of independent interest.

	\subsection{Merging BWTs in Small Space}\label{sec:algo2}
	
	The procedure of Algorithm \ref{alg:fill leaves} can be extended to merge BWTs of two collections $\mathcal C_1$, $\mathcal C_2$ using $o(n\log\sigma)$ bits of working space on top of the input BWTs and output Document Array (here, $n$ is the cumulative length of the two BWTs). The idea is to simulate a navigation of the leaves of the generalized suffix tree of $\mathcal C_1 \cup \mathcal C_2$ (note: for us, a collection is an ordered multi-set of strings). Each leaf is represented by a pair of intervals, respectively on $BWT(\mathcal C_1)$ and $BWT(\mathcal C_2)$, of strings of the form $W\#$. Note that: (i) the suffix array of $\mathcal C_1 \cup \mathcal C_2$ is covered by the non-overlapping intervals of strings of the form $W\#$, and (ii) for each such string $W\#$, the interval $\mathtt{range(W\#)} = \langle L,R \rangle$ in $GSA(\mathcal C_1 \cup \mathcal C_2)$ can be partitioned as $\langle L, M \rangle \cdot \langle M+1, R\rangle$, where $\langle L,M\rangle$ contains only suffixes from  $\mathcal C_1$ and $\langle M+1,R \rangle$ contains only suffixes from  $\mathcal C_2$ (one of these two intervals could be empty).	
	It follows that we can navigate in parallel the leaves of the suffix trees of $\mathcal C_1$ and $\mathcal C_2$ (using again a stack or a queue containing pairs of intervals on the two BWTs), and fill the Document Array $DA[1..n]$, an array that will tell us whether the $i$-th entry of $BWT(\mathcal C_1 \cup \mathcal C_2)$ comes from $BWT(\mathcal C_1)$ ($DA[i] = 0$) or $BWT(\mathcal C_2)$ ($DA[i] = 1$). To do this, let $\langle L_1, R_1\rangle$ and $\langle L_2, R_2\rangle$ be the ranges on the suffix arrays of $\mathcal C_1$ and $\mathcal C_2$, respectively, of a  suffix $W\#$ of some string in the collections. 
	Note that one of the two intervals could be empty: $R_j<L_j$. In this case, we still require that $L_j-1$ is the number of suffixes in $\mathcal C_j$ that are smaller than $W\#$.
	Then, in the collection $\mathcal C_1 \cup \mathcal C_2$ there are $L_1 + L_2 - 2$ suffixes smaller than $W\#$, and $R_1 + R_2$ suffixes smaller than or equal to $W\#$. It follows that the range of $W\#$ in the suffix array of $\mathcal C_1 \cup \mathcal C_2$ is $\langle L_1+L_2-1, R_1+R_2\rangle$, where the first $R_1-L_1+1$ entries correspond to suffixes of strings from $\mathcal C_1$. Then, we set $DA[L_1+L_2-1, L_2 + R_1-1] \leftarrow 0$ and $DA[L_2 + R_1,R_1+R_2] \leftarrow 1$. 
	The procedure starts from the pair of intervals corresponding to the ranges of the string ``$\#$'' in the two BWTs, and proceeds recursively by left-extending the current pair of ranges $\langle L_1, R_1\rangle$, $\langle L_2, R_2\rangle$ with the symbols in $\mathtt{BWT_1.range\_distinct(L_1,R_1)} \cup \mathtt{BWT_2.range\_distinct(L_2,R_2)}$.
	For space reasons, the detailed procedure is reported in Appendix \ref{app:merge} as Algorithm \ref{alg:merge}. The leaf visit is implemented, again, using a stack or a queue; this time however, these containers are filled with pairs of intervals $\langle L_1, R_1\rangle$, $\langle L_2, R_2\rangle$. 
	We implement the stack simply as a vector of quadruples $\langle L_1, R_1, L_2, R_2\rangle$. As far as the queue is concerned, some care needs to be taken when representing the pairs of ranges using bitvectors as seen in Appendix \ref{app:beller} with Beller et al.'s representation. 
	Recall that, at any time, the queue can be partitioned in two sub-sequences associated with LCP values $\ell$ and $\ell+1$ (we pop from the former, and push in the latter).
	This time, we represent each of these two subsequences as a vector of quadruples (pairs of ranges on the two BWTs) as long as the number of quadruples in the sequence does not exceed $n/\log n$. When there are more quadruples than this threshold, we switch to a bitvector representation defined as follows. 
	Let $|BWT(\mathcal C_1)|=n_1$, $|BWT(\mathcal C_2)|=n_2$, and $|BWT(\mathcal C_1\cup \mathcal C_2)| = n = n_1+n_2$.
	We keep two bitvectors $\mathtt{Open[1..n]}$ and $\mathtt{Close[1..n]}$ storing opening and closing parentheses of intervals in $BWT(\mathcal C_1\cup \mathcal C_2)$. We moreover keep two bitvectors $\mathtt{NonEmpty_1[1..n]}$ and $\mathtt{NonEmpty_2[1..n]}$ keeping track, for each $i$ such that $\mathtt{Open[i]=1}$, of whether the interval starting in $BWT(\mathcal C_1\cup \mathcal C_2)[i]$ contains suffixes of reads coming from $\mathcal C_1$ and $\mathcal C_2$, respectively. Finally, we keep four bitvectors $\mathtt{Open_j[1..n_j]}$ and $\mathtt{Close_j[1..n_j]}$, for $j=1,2$, storing non-empty intervals on $BWT(\mathcal C_1)$ and $BWT(\mathcal C_2)$, respectively. To insert a pair of intervals $\langle L_1, R_1\rangle,\ \langle L_2, R_2\rangle$ in the queue, let $\langle L,R \rangle = \langle L_1+L_2-1, R_1+R_2\rangle$. We set $\mathtt{Open[L]} \leftarrow 1$ and  $\mathtt{Close[R]} \leftarrow 1$. Then, for $j=1,2$, we set $\mathtt{NonEmpty_j[L]} \leftarrow 1$, $\mathtt{Open_j[L_j]} \leftarrow 1$ and $\mathtt{Close_j[R_j]} \leftarrow 1$ if and only if $R_j\geq L_j$. 
	This queue representation takes $O(n)$ bits. 
	By construction, for each bit set in $\mathtt{Open}$ at position $i$, there is a corresponding bit set in 
	$\mathtt{Open_j}$ if and only if $\mathtt{NonEmpty_j[i]} = 1$ (moreover, corresponding bits set appear in the same order in $\mathtt{Open}$ and
	$\mathtt{Open_j}$). It follows that a left-to-right scan of these bitvectors is sufficient to identify corresponding intervals on $BWT(\mathcal C_1\cup \mathcal C_2)$, $BWT(\mathcal C_1)$, and $BWT(\mathcal C_2)$.
	By packing the bits of the bitvectors in words of $\Theta(\log n)$ bits, the $t$ pairs of intervals contained in the queue can be extracted in $O(t+ n/\log n)$ time (as described in~\cite{beller2013computing}) by scanning in parallel the bitvectors forming the queue. Particular care needs to be taken only when we find the beginning of an interval $\mathtt{Open[L]=1}$ with $\mathtt{NonEmpty_1[L]} = 0$ (the case $\mathtt{NonEmpty_2[L]} = 0$ is symmetric). Let $L_2$ be the beginning of the corresponding non-empty interval on $BWT(\mathcal C_2)$. Even though we are not storing $L_1$ (because we only store nonempty intervals), we can retrieve this value as $L_1=L-L_2+1$. Then, the empty interval on $BWT(\mathcal C_1)$ is $\langle L_1, L_1-1\rangle$.
	
	The same arguments used in the previous section show that the algorithm runs in $O(n\log\sigma)$ time and uses $o(n\log\sigma)$ bits of space on top of the input BWTs and output Document Array. This proves Theorem \ref{th:merge}.
	%At the end of execution, the merged BWT can be directly streamed to output by scanning in parallel the Document Array and the input BWTs. 
	To conclude, we note that the algorithm can be extended to compute the LCP array of the merged collection while merging the BWTs. This requires adapting Algorithm \ref{alg:fill nodes} to work on pairs of suffix tree nodes (as we did in Algorithm \ref{alg:merge} with pairs of leaves), but for space reasons we do not describe all details here. Results on an implementation of the extended algorithm are discussed in the next section. 
	From the practical point of view, note that it is more advantageous to induce the LCP of the merged collection while merging the BWTs (rather than first merging and then inducing the LCP using the algorithm of the previous section), since leaf-type LCP values can be induced directly while computing the document array.

	\section{Implementation and Experimental Evaluation}\label{sec:experiments}
	
	We implemented our algorithms on DNA alphabet in \repo using the language C++. 
	Thanks to the small alphabet size, it was actually sufficient to implement our extension of Belazzougui's enumeration algorithm (and not the strategy of Beller et al., which is more suited to large alphabets).
	The repository features a new packed string on DNA alphabet $\Sigma_{DNA}=\{A,C,G,T,\#\}$ using 4 bits per character and able to compute the quintuple $\langle rank_c(i) \rangle_{i\in \Sigma_{DNA}}$ with just one cache miss.  This is crucial for our algorithms, since at each step we need to left-extend ranges by all characters. We also implemented a packed string on the augmented alphabet $\Sigma_{DNA}^+=\{A,C,G,N,T,\#\}$ using $4.38$ bits per character and offering the same cache-efficiency guarantees. Several heuristics have been implemented to reduce the number of cache misses in practice. In particular, we note that in Algorithm \ref{alg:fill leaves} we can avoid backtracking when the range size becomes equal to one; the same optimization can be implemented in Algorithm \ref{alg:merge} when also computing the LCP array, since leaves of size one can be identified during navigation of internal suffix tree nodes. Overall, we observed (using a memory profiler) that in practice the combination of Algorithms \ref{alg:fill nodes}-\ref{alg:fill leaves} generates at most $1.5n$ cache misses; the extension of Algorithm \ref{alg:merge} that computes also LCP values generates twice this number of cache misses (this is expected, since it navigates two BWTs).
	
	We now report some preliminary experiments on our algorithms: \induceLCP (Algorithms \ref{alg:fill nodes}-\ref{alg:fill leaves}) and \mergeBwtLCP (Algorithm \ref{alg:merge}, extended to compute also the LCP array). 
	All tests were done on a DELL PowerEdge R630 machine, used in non exclusive mode.
	Our platform is a $24$-core machine with Intel(R) Xeon(R) CPU E5-2620 v3 at $2.40$ GHz, with $128$ GiB of shared memory. The system is Ubuntu 14.04.2 LTS.
	%To assess the performance of our algorithms, we consider two experiments.
	
	\begin{table}[t]
		\centering
		\begin{tabular}{|@{\ }l@{\ }|@{\ }c@{\ }|@{\ }c@{\ }|@{\ }c@{\ }|@{\ }c@{\ }|c@{\ }|c@{\ }|}
			\hline
			Name        & Size & $\sigma$          & N. of  & Max      & Bytes for\\
			& GiB   &              & reads  & length   & lcp values\\
			\hline
			NA12891.8           & 8.16      & 5               & 85,899,345  & 100 & 1   \\
			\hline
			shortreads         & 8.0     & 6               & 85,899,345  & 100      & 1  \\
			\hline
			pacbio                & 8.0      & 6            & 942,248      & 71,561  & 4   \\
			\hline
			pacbio.1000       & 8.0      & 6              & 8,589,934    & 1000    & 2\\
			\hline
				NA12891.24        &   23.75    & 6               & 250,000,000  & 100 & 1   \\
			\hline
			NA12878.24           &  23.75    & 6               & 250,000,000  & 100 & 1   \\
			\hline
		\end{tabular}
		\caption{Datasets used in our experiments. Size accounts only for the alphabet's characters. The alphabet's size $\sigma$ includes the terminator.\vspace{-20pt}}
		\label{tableDataset}
	\end{table}
	
	Table \ref{tableDataset} summarizes the datasets used in our experiments. 
	``NA12891.8G'' contains Human DNA reads on the alphabet $\Sigma_{DNA}$ downloaded from~\footnote{\url{ftp://ftp.1000genomes.ebi.ac.uk/vol1/ftp/phase3/data/NA12891/sequence_read/SRR622458_1.filt.fastq.gz}}, where we have removed reads containing the nucleotide $N$.
	``shortreads'' contains Human DNA short reads on the extended alphabet $\Sigma_{DNA}^+$. 
	%downloaded from~\footnote{\url{ftp://ftp.sra.ebi.ac.uk/vol1/ERA015/ERA015743/srf/}}.
	%trimmed to length 100. 
	%, downloaded from\footnote{\url{https://trace.ncbi.nlm.nih.gov/Traces/sra/?run=ERR1942989}}.
	%trimmed to length 300. 
	``pacbio'' contains PacBio RS II reads from the species \emph{Triticum aestivum} (wheat).
	%, downloaded from~\footnote{\url{https://trace.ncbi.nlm.nih.gov/Traces/sra/?run=SRR5816161}}.
	% with different lengths. 
	``pacbio.1000'' are the strings from ``pacbio'' trimmed to length 1,000. 
	All the above datasets except the first have been download from \url{https://github.com/felipelouza/egap/tree/master/dataset}.
	To conclude, we added two collections, ``NA12891.24'' and ``NA12878.24'' obtained by taking the first $250,000,000$ reads from individuals NA12878\footnote{\url{ftp://ftp.1000genomes.ebi.ac.uk/vol1/ftp/phase3/data/NA12878/sequence_read/SRR622457_1.filt.fastq.gz}} and NA12891. All datasets except ``NA12891.8'' are on the alphabet $\Sigma_{DNA}^+$. In Tables \ref{tab:merge} and \ref{tab:induce}, the suffix ``.RC'' added to a dataset's name indicates the reverse-complemented dataset.

	We  compare our algorithms with \eGap\footnote{\url{https://github.com/felipelouza/egap}} and BCR~\footnote{\url{https://github.com/giovannarosone/BCR_LCP_GSA}}, two tools designed to build the BWT and LCP of a set of DNA reads. 
	Since no tools for inducing the LCP from the BWT of a set of strings are available in the literature, in Table \ref{tab:induce} we simply compare the resources used by \induceLCP with the time and space requirements of  \eGap and BCR when building the BWT.
	In \cite{egidiLouzaManziniTellesWABI2018}, experimental results show that BCR works better on short reads and collections with a large average LCP, while \eGap works better when the datasets contain long reads and relatively small average LCP. 
	For this reason, in the preprocessing step we have used BCR for the collections containing short reads and \eGap for the other collections.
	\eGap, in addition, is capable of merging two or more BWTs while inducing the LCP of their union. In this case, we can therefore directly compare the performance of \eGap with our tool \mergeBwtLCP; results are reported in Table \ref{tab:merge}.
	Since the available RAM is greater than the size of the input, we have used the semi-external strategy of \eGap.
	Notice that an entirely like-for-like comparison between our tools and \eGap is not completely feasible, being \eGap a semi-external memory tool (our tools, instead, use internal memory only).
	While in our tables we report RAM usage only, it is worth to notice that \eGap uses a considerable amount of disk working space. For example, the tool uses $56$GiB of disk working space when run on a $8$GiB input (in general, the disk usage is of $7n$ bytes).
	
	As predicted by theory, our tools exhibit a dataset-independent linear time complexity (whereas \eGap requires more processing time on datasets with long average LCP).
	Table \ref{tab:induce} shows that our tool \induceLCP induces the LCP from the BWT faster than building the BWT itself. When 'N's are not present in the dataset, \induceLCP processes data at a rate of $2.92$ megabases per second and uses $0.5$ Bytes per base in RAM in addition to the LCP. When 'N's are present, the throughput decreases to $2.12$ megabases per second and the tool uses  $0.55$ Bytes per base in addition to the LCP.
	As shown in Table \ref{tab:merge}, our tool \mergeBwtLCP is from $1.25$ to $4.5$ times faster than \eGap on inputs with large average LCP, but $1.6$ times slower when the average LCP is small (dataset ``pacbio''). When 'N's are not present in the dataset, \mergeBwtLCP processes data at a rate of $1.48$ megabases per second and uses $0.625$ Bytes per base in addition to the LCP. When 'N's are present, the throughput ranges from $1.03$ to $1.32$ megabases per second and the tool uses  $0.673$  Bytes per base in addition to the LCP. When only computing the merged BWT (results not shown here for space reasons), \mergeBwtLCP uses in total $0.625$/$0.673$ Bytes per base in RAM (without/with 'N's) and is about $1.2$ times faster than the version computing also the LCP.

	\begin{table}[t]
		\centering
		\begin{tabular}{|c|c|c|c|c|c|c|c|c|}
			\hline
			& \multicolumn{2}{c|}{Preprocessing}     & \multicolumn{2}{c|}{\eGap}    & \multicolumn{2}{c|}{\mergeBwtLCP}    \\ 
			\hline
			Name            &  Wall Clock &   RAM   &   Wall Clock    & RAM                & Wall Clock   & RAM      \\ 
			&  (h:mm:ss)  &    (GiB) &   (h:mm:ss)    &  (GiB)              &  (h:mm:ss)   &  (GiB)     \\
			\hline
			NA12891.8    & 1:15:57 & 2.84 & \multirow{2}{*}{10:15:07}  &   \multirow{2}{*}{18.09 (-m 32000)}    &   \multirow{2}{*}{3:16:40}    &  \multirow{2}{*}{26.52}  \\
			\cline{1-3}
			NA12891.8.RC & 1:17:55 & 2.84 &                            &                                        &                               &             \\
			\hline
			shortreads     &  1:14:51 & 2.84 & \multirow{2}{*}{11:03:10}  &    \multirow{2}{*}{16.24  (-m 29000)}  &   \multirow{2}{*}{3:36:21}    &  \multirow{2}{*}{26.75}   \\ 
			\cline{1-3}
			shortreads.RC  & 1:19:30 & 2.84 &                            &                                        &                               &             \\
			\hline
			pacbio.1000    &  2:08:56 & 31.28 & \multirow{2}{*}{5:03:01}   &   \multirow{2}{*}{21.23 (-m 45000)}    &  \multirow{2}{*}{4:03:07}     &  \multirow{2}{*}{42.75}   \\ 
			\cline{1-3}
			pacbio.1000.RC &  2:15:08 & 31.28 &                            &                                        &                               &             \\
			\hline
			pacbio         &  2:27:08 & 31.25 & \multirow{2}{*}{2:56:31}   &   \multirow{2}{*}{33.40 (-m 80000)}    &   \multirow{2}{*}{4:38:27}    &  \multirow{2}{*}{74.76}        \\ 
			\cline{1-3}
			pacbio.RC      &  2:19:27 & 31.25 &                            &                                        &                               &             \\
			\hline
				NA12878.24   & 4:24:27  & 7.69  & \multirow{2}{*}{31:12:28}   &   \multirow{2}{*}{47.50 (-m 84000)}    &   \multirow{2}{*}{6:41:35}    &  \multirow{2}{*}{73.48}        \\ 
			\cline{1-3}
			NA12891.24      & 4:02:42  & 7.69  &                            &                                        &                               &             \\
			\hline
		\end{tabular}
		\caption{In this experiment, we merge pairs of BWTs and induce the LCP of their union using \eGap and \mergeBwtLCP. We also show the resources used by the pre-processing step (building the BWTs) for comparison. Wall clock is the elapsed time from start to completion of the instance, while RAM (in GiB) is the peak Resident Set Size (RSS). All values were taken using the \texttt{/usr/bin/time} command. During the preprocessing step on the collections pacBio.1000 and pacBio, the available memory in MB (parameter m) of \eGap was set to 32000 MB. In the merge step this parameter was set to about to the memory used by \mergeBwtLCP.\vspace{-15pt}}
		\label{tab:merge}
	\end{table}

	%\eGAP sui shortRead
	%NA12891.8GiB    &  1:22:17 & 32.77 
	%NA12891.8GiB RC &  1:13:11 & 32.77 
	%shortreads     &  3:24:33 & 32.77
	%shortreads RC  &  3:21:44 & 32.77
	
	%\todo{Nota che conviene il merge BWT-LCP se le bwt delle due collezioni si calcolono in parallelo, altrimenti i tempi diciamo che sono simili.}
	
	%Table \ref{tab:induce} shows
	
	%./induce_lcp -i merge_opz.N_SRR622458_1.filt.85899345reads_8GiB_senzaRigheN_F+RC.bwt -o inducemerge_opz.N_SRR622458_1.filt.85899345reads_8GiB_senzaRigheN_F+RC.out.lcp -l 1 -n -t \0
	
	%nohup /usr/bin/time -v ./induce_lcp -i shortreads.100.F_RC_out.bwt -o induceshortreads.100.F_RC_out.lcp -l 1 -t '\0' > res_induce_shortreads.100_F+RC_1byte_2019-01-03.stdout 2> res_induce_shortreads.100_F+RC_1byte_2019-01-03.stderr &
	
	%nohup /usr/bin/time -v ./induce_lcp -i merge_pacbio.1000_F+RC.bwt -o inducemerge_pacbio.1000_F+RC.out.lcp -l 2 -t '\0' > res_induce_pacbio.1000_F+RC_2byte_2019-01-07.stdout 2> res_induce_pacbio.1000_F+RC_2byte_2019-01-07.stderr &
	
	%nohup /usr/bin/time -v ./induce_lcp -i merge_pacbio_F+RC.bwt.bwt -o inducemerge_pacbio_F+RC_out.lcp -l 4 -t '\0' > res_induce_pacbio_F+RC_1byte_2019-01-07.stdout 2> res_induce_pacbio_F+RC_1byte_2019-01-07.stderr &
	
	\begin{table}[t]
		\centering
		\begin{tabular}{|c|c|c|c|c|c|c|}
			\hline
			& \multicolumn{2}{c|}{Preprocessing}  & \multicolumn{2}{c|}{\induceLCP}          \\ 
			\hline
			Name         & Wall Clock    & RAM       & Wall Clock  & RAM      \\ 
			&  (h:mm:ss)    & GiB        &  (h:mm:ss)  &  (GiB)    \\
			\hline
			NA12891.8 $\cup$ NA12891.8.RC (BCR)  &  2:43:02      & 5.67   &  1:40:01 &    24.48     \\
			\hline
			shortread $\cup$ shortread.RC (BCR) & 2:47:07    &   5.67   &  2:14:41 &   24.75     \\ 
			\hline
			pacbio.1000 $\cup$ pacbio.1000.RC (\eGap -m 32000)  &     7:07:46   &  31.28  &     1:54:56    &   40.75         \\ 
			\hline
			pacbio $\cup$ pacbio.RC (\eGap -m 80000)      &  6:02:37     &    78.125       &   2:14:37   &   72.76    \\ 
			\hline
			NA12878.24 $\cup$ NA12891.24 (BCR)      & 8:26:34      &  16.63        &   6:41:35   &   73.48    \\ 
			\hline
		\end{tabular}
		\caption{In this experiment, we induced the LCP array from the BWT of a collection (each collection is the union of two collections from Table \ref{tab:merge}). We also show pre-processing requirements (i.e. building the BWT) of the better performing tool between BCR and \eGap.\vspace{-15pt}}
		\label{tab:induce}
	\end{table}

	%EGAP: shortreads   &   12:25:21    &   32.70
	
	%BCR (by keeping the eBWT in internal memory) builds the eBWT of the dataset NA12891 and the dataset NA12878 in about 4 hours by using 44 GiB of RAM each.
	%Moreover, BCR (by keeping the eBWT in internal memory) builds the eBWT (about $50$ GiB) of the two collections in about $9$ hours by using $87$GiB of RAM.
	%Table \ref{NA12878andNA12891} shows the performance of \eGap and our algorithms.

	%nohup /usr/bin/time -v ./merge_bwt -1 NA12878_SRR622457_1.filt.250Mreads.fasta.out.ebwt -2 SRR622458_1.filt.250Mreads.fasta.out.ebwt -o merge_SRR622457_SRR622458_1 -l 1  > res_two_bwtmerge_NA12878_NA12_1byte_2018-12-24.stdout 2> res_two_bwtmerge_NA12878_NA12_1byte_2018-12-24.stderr &
	
	%nohup /usr/bin/time -v ./eGap -b --lbytes 1 --lcp -v -m 84000 -o eGap_NA12878_NA12891 dataset/SRR622457_1.filt.250Mreads.ebwt dataset/SRR622458_1.filt.250Mreads.ebwt > res_eGap_NA12878_NA12891_2018-12-30.stdout 2> res_eGap_NA12878_NA12891_2018-12-30.stderr &

	\bibliographystyle{plain}
	\bibliography{paper}
	
	\newpage
	
	\appendix
	
		\section{Basic Concepts}\label{sec:notation}
	
	Let $\Sigma =\{c_1, c_2, \ldots, c_\sigma\}$ be a finite ordered alphabet of size $\sigma$ with $c_1< c_2< \ldots < c_\sigma$, where $<$ denotes the standard lexicographic order. 
	%We denote by $\Sigma^*$ the set of strings over $\Sigma$. 
	 
	Given a text $T=t_1 t_2 \cdots t_n \in \Sigma^*$ we denote by $|T|$ its length $n$. 
	We use $\epsilon$ to denote the empty string.
	A \emph{factor} (or \emph{substring}) of $T$ is written as $T[i,j] = t_i \cdots t_j$ with $1\leq i \leq j \leq n$.  
	%A factor of type $w[1,j]$ is called a \emph{prefix}, while a factor of type $w[i,n]$ is called a \emph{suffix}. We also use the notation $w[i..]$ to denote the suffix $w[i,n]$. 
	%The $i$-th symbol in $T$ is denoted by $T[i]$.  
	%Given a text $T$ and $c\in\Sigma$, we write $T.\rank(c,i)$ to denote the number of occurrences of $c$ in $T[1,i]$, and $T.\sel(c,j)$ to denote the position of the $j$-th $c$ in~$T$.  \todo{in realta' bisogna scriverla secondo la tua notazione.}
	
	A \emph{right-maximal} substring $W$ of $T$ is a string for which there exist at least two distinct characters $a,b$ such that $Wa$ and $Wb$ occur in $T$. 
	
	With $\mathcal C = \{T_1, \dots, T_m\}$ we denote a string collection of total length $n$, where each $T_i$ is terminated by a character $\#$ (the terminator) lexicographically smaller than all other alphabet's characters. In particular, a collection is an ordered multiset, and we denote $\mathcal C[i] = T_i$.
	
	The \emph{working space} of an algorithm is the total space used during computation in addition to the input and the output. We moreover assume that input and output are re-writable. 
	
	The \emph{generalized suffix array} $GSA[1..n]$ (see~\cite{Shi:1996,CGRS_JDA_2016,Louza2017}) of  $\mathcal C$ is an array of pairs $GSA[i] = \langle j,k \rangle$ such that $\mathcal C[j][k..]$ is the $i$-th lexicographically smallest suffix of strings in $\mathcal C$, where we break ties by input position (i.e. $j$ in the notation above).
	
	We denote by $\mathtt{range(W)} = \langle \mathtt{left(W)}, \mathtt{right(W)} \rangle$ the maximal pair $\langle L,R \rangle$ such that all suffixes in $GSA[L,R]$ are prefixed by $W$. 
	Note that the number of suffixes lexicographically smaller than $W$ in the collection is $L-1$. 
	We extend this definition also to cases where $W$ is not present in the collection: in this case, the (empty) range is $\langle L, L-1\rangle$ and we still require that $L-1$ is the number of suffixes lexicographically smaller than $W$ in the collection. 
	 
	The \emph{extended Burrows-Wheeler Transform}
	$BWT[1..n]$ \cite{MantaciRRS07,BauerCoxRosoneTCS2013} of $\mathcal C$ is the character array defined as $BWT[i] = \mathcal C[j][k-1\ \mathtt{mod}\ |\mathcal C[j]|]$, where $\langle j,k \rangle =  GSA[i]$. 
	
	%\todo{Penso dovremmo distinguere la stringa bwt (di 1 stringa) con ebwt (di 1 collezione). Sono diverse, la prima contiene un terminatore in piu'. Ma se hai una stringa, usi $\#$ come terminatore, no?}
	The \emph{longest common prefix} (LCP) array of a collection $\mathcal C$ of strings (see \cite{CGRS_JDA_2016,Louza2017,egidiLouzaManziniTellesWABI2018}) is an array storing the length of the longest common prefixes between two consecutive suffixes of $\mathcal C$ in lexicographic order (with $LCP[1]=0$). 
	
	Given two collections $\mathcal C_1, \mathcal C_2$ of total length $n$, the Document Array of their union is the binary array $DA[1..n]$ such that $DA[i] = 0$ if and only if the $i$-th smallest suffix comes from $\mathcal C_1$. When merging suffixes of the two collections, ties are broken by collection number (i.e. suffixes of $\mathcal C_1$ are smaller than suffixes of $\mathcal C_2$ in case of ties).
	
	$S.rank_c(i)$ is the number of characters equal to $c$ in $S[1,i-1]$. 
	
	The $C$-array of a string (or collection) $S$ is an array $C[1..\sigma]$ such that $C[i]$ contains the number of characters lexicographically smaller than $i$ in $S$, plus one ($S$ will be clear from the context). 
	Alternatively, $C[c]$ is the starting position of suffixes starting with $c$ in the suffix array of the string. 
	When $S$ (or any of its permutations) is represented with a balanced wavelet tree, then we do not need to store explicitly $C$, and $C[c]$ can be computed in $O(\log\sigma)$ time with no space overhead on top of the wavelet tree~\cite{navarro2012wavelet}. In the rest of the paper we assume $C$ is accessed in this way. 
	
	Function $\mathtt{getIntervals(L,R,BWT)}$, where $BWT$ is the Burrows-Wheeler transform of a string collection and $\langle L,R\rangle$ is the suffix array interval of some string $W$ appearing in the collection, returns all suffix array intervals of strings $cW$, with $c\neq \#$, that occur in the collection. 
	When $BWT$ is represented with a balanced wavelet tree, we can implement this function so that it terminates in $O(\log\sigma)$ time per returned interval~\cite{beller2013computing}. Importantly, the function can be made to return the output intervals on-the-fly, one by one (in an arbitrary order), without the need to store them all in an auxiliary vector, with just $O(\log n)$ bits of additional overhead in space~\cite{beller2013computing} (essentially, this requires to DFS-visit the sub-tree of the wavelet tree induced by $BWT[L,R]$; the visit requires only $\log\sigma$ bits to store the current path in the tree).
	
	We note that an extension of the above function that navigates in parallel two BWTs is immediate. Function $\mathtt{getIntervals(L_1,R_1,L_2, R_2, BWT_1, BWT_2)}$ takes as input two ranges of a string $W$ on the BWTs of two collections, and returns the pairs of ranges on the two BWTs corresponding to all left-extensions $cW$ of $W$ ($c\neq \#$) such that $cW$ appears in at least one of the two collections. To implement this function, it is sufficient to navigate in parallel the two wavelet trees as long as at least one of the two intervals is not empty.
	
	The function $S.\mathtt{range\_distinct(i,j)}$ returns the set of distinct alphabet characters \emph{different than the terminator} $\#$ in $S[i,j]$. Also this function can be implemented in $O(\log\sigma)$ time per returned element when $S$ is represented with a wavelet tree (again, this requires a DFS-visit of the sub-tree of the wavelet tree induced by $S[i,j]$).
	
	$BWT.\mathtt{bwsearch(\langle L,R \rangle, c)}$ is the procedure that, given the suffix array interval $\langle L,R \rangle$ of a string $W$, returns the suffix array interval of $cW$ by using the BWT~\cite{ferragina2000opportunistic}. This function requires access to array $C$ and \emph{rank} support on $BWT$ (read above), and runs in $O(\log\sigma)$ time when $BWT$ is represented with a balanced wavelet tree. 
	
	\section{Notes on Belazzougui's Algorithm}\label{app:belazzougui}
	
	As discussed in the main paper, the enumeration algorithm works by visiting the Weiner link tree of the text.
	While this guarantees that we will visit all and only the suffix tree's explicit nodes (for all details, see~\cite{belazzougui2014linear}), there are two main issues that need to be addressed. First, the stack's size may grow in an uncontrolled way. The solution to this problem is simple: once computed $\mathtt{repr(cW)}$ for the right-maximal left-extensions $cW$ of $W$, we push them on the stack in decreasing order of range length $range(cW)$ (i.e. the node with the smallest range is pushed last). This guarantees that the stack will always contain at most $O(\sigma\log n)$ elements~\cite{belazzougui2014linear}. Since each element takes $O(\sigma\log n)$ bits to be represented, the stack's size never exceeds $O(\sigma^2\log^2 n)$ bits. 
	
	The second issue that needs to be addressed is how to efficiently compute $\mathtt{repr(cW)}$ from $\mathtt{repr(W)}$ for the characters $c$ such that $cW$ is right-maximal in $T$. In~\cite{belazzougui2014linear,belazzougui2016linear}  this operation is supported efficiently by first enumerating all \emph{distinct} characters in each range $BWT[\mathtt{first_W[i]..  first_W[i+1]}]$ for $i=1, \dots, k_W$. 
	Using the notation of \cite{belazzougui2014linear}, let us call $\mathtt{rangeDistinct(i,j)}$ the operation that returns all distinct characters in $BWT[i,j]$.
	Equivalently, for each $a\in \mathtt{chars_W}$ we want to list all distinct left-extensions $cWa$ of $Wa$.
	Note that, in this way, we may also visit implicit suffix tree nodes (i.e. some of these left-extensions could be not right-maximal).
	Stated otherwise, we are traversing all explicit \emph{and} implicit Weiner links. Since the number of such links is linear~\cite{belazzougui2014linear,belazzougui2014alphabet} (even including implicit Weiner links\footnote{To see this, first note that the number of right-extensions $Wa$ of $W$ that have only one left-extension $cWa$ is at most equal to the number of right-extensions of $W$; globally, this is at most the number of suffix tree's nodes (linear). Any other right-extension $Wa$ that has at least two distinct left-extensions $cWa$ and $bWa$ is, by definition, left maximal and corresponds therefore to a node in the suffix tree of the reverse of $T$. It follows that all left-extensions of $Wa$ can be charged to an edge of the suffix tree of the reverse of $T$ (again, the number of such edges is linear).}), globally the number of distinct characters returned by $\mathtt{rangeDistinct}$ operations is $O(n)$. An implementation of $\mathtt{rangeDistinct}$ on wavelet trees is discussed in \cite{beller2013computing} with the procedure \texttt{getIntervals} (this procedure actually returns more information: the suffix array range of each $cWa$). This implementation runs in $O(\log\sigma)$ time per returned character. Globally, we therefore spend $O(n\log\sigma)$ time using a wavelet tree. At this point, we need to compute the representation $\mathtt{repr(cW)}$ for all left-extensions of $W$ and keep only the right-maximal ones. Letting $x=\mathtt{repr(W)}$, we call $\mathtt{BWT.Weiner(x)}$ the function that returns the representations of such strings (this function will be used in Line \ref{range distinct2} of Algorithm \ref{alg:fill nodes}).
	This function can be implemented by observing that 
	$$
	\mathtt{range(cWa) = \langle\ C[c] + BWT.rank_c(left(Wa)), C[c] + BWT.rank_c(right(Wa)+1)-1 \ \rangle} 
	$$
	where $a=\mathtt{chars_W[i]}$ for $1\leq i < |\mathtt{first_W}|$, and noting that $\mathtt{left(Wa)}$ and $\mathtt{right(Wa)}$ are available in $\mathtt{repr(W)}$. Note also that we do not actually need to know the value of characters $\mathtt{chars_W[i]}$ to compute the ranges of each $cW\cdot \mathtt{chars_W[i]}$; this is the reason why we can omit $\mathtt{chars_W}$ from $\mathtt{repr(W)}$.
	Using a wavelet tree, the above operation takes $O(\log\sigma)$ time. By the above observations, the number of strings $cWa$ such that $W$ is right-maximal is bounded by $O(n)$. Overall, computing $\mathtt{repr(cW)} = \langle \mathtt{first_{cW}}, |W|+1 \rangle$ for all left-extensions $cW$ of all right-maximal strings $W$ takes therefore $O(n\log\sigma)$ time. Within the same running time, we can check which of those extensions is right maximal (i.e. those such that $|\mathtt{first_{cW}}|\geq 3$), sort them by interval length (we always sort at most $\sigma$ node representations, therefore also sorting takes globally $O(n\log\sigma)$ time), and push them on the stack.
	
	\section{Notes on Beller et al.'s Algorithm}\label{app:beller}
	
	\textbf{Time complexity} It is easy to see that the algorithm inserts in total a linear number of intervals in the queue since an interval $\langle L_i, R_i, \ell+1 \rangle$ is inserted only if $LCP[R_i+1] = \bot$, and successively $LCP[R_i+1]$ is set to a value different than $\bot$. Clearly, this can happen at most $n$ times. 
	In~\cite{beller2013computing} the authors moreover show that, even counting the left-extensions of those intervals (that we compute after popping each interval from the queue), the total number of computed intervals stays linear.
	Overall, the algorithm runs therefore in $O(n\log\sigma)$ time (as discussed in the Appendix \ref{sec:notation}, $\mathtt{getIntervals}$ runs in $O(\log\sigma)$ time per returned element, see also Appendix~\ref{beller et al.}).
	
	\textbf{Queue implementation} To limit space usage, Beller et al. use the following queue representations. First note that, at each time point, the queue's triples are partitioned in a (possibly empty) sequence with associated LCP value (i.e. the third element in the triples) $\ell+1$, followed by a sequence with associated LCP value $\ell$, for some $\ell$. We can therefore store the two sequences (with associated LCP value) independently, and there is no need to store the LCP values in the triples themselves (i.e. the queue's elements become just ranges). Note also that we pop elements from the sequence with associated LCP value $\ell$, and push elements in the sequence with associated LCP value $\ell+1$. 
	When the former sequence is empty, we create a new sequence with associated LCP value $\ell+2$ and start popping from the sequence with associated LCP value $\ell+1$ (and so on).
	Beller et al. represent each of the two sequences separately as follows. While inserting elements in a sequence, as long as the sequence's length does not exceed $n/\log n$ we represent it as a vector of pairs (of total size at most $O(n)$ bits). This representation supports push/pop operations in (amortized) constant time. 
	As soon as the sequence's length exceeds $n/\log n$, we switch to a representation that uses two packed bitvectors of length $n$ storing, respectively, the left- and right-most boundaries of the ranges in the sequence. 
	Note that this representation can be used because the sequence of intervals corresponds to suffix array ranges of strings of some fixed length $\ell$, therefore there cannot be overlapping intervals.
	Pushing an interval in this new queue's representation takes constant time. Popping all the $t$ intervals from one of the two sequences, on the other hand, can be implemented in $O(t+ n/\log n)$ time by scanning the bitvectors (this requires using simple bitwise operations on words, see~\cite{beller2013computing} for all details). Since at  most $O(\log n)$ sequences will exceed size $n/\log n$, overall pop operations take amortized constant time.

	\section{Proofs}\label{app:proof}

	\begin{lemma}\label{lemma:proof of thm1}
		Algorithms \ref{alg:fill nodes} and \ref{alg:fill leaves} correctly compute the LCP array of the collection in $O(n\log\sigma)$ time using $o(n\log\sigma)$ bits of working space on top of the input and output.
	\end{lemma}
	\begin{proof}
		
		\textbf{Correctness and completeness - Algorithm \ref{alg:fill nodes}}. We start by proving correctness and completeness of Beller et al.'s procedure in Line \ref{beller et al.} of Algorithm \ref{alg:fill nodes} (procedure \texttt{BGOS(BWT,LCP)}). The proof that the algorithm fills all the node-type LCP entries correctly proceeds by induction on the LCP value $\ell$ and follows the original proof of~\cite{beller2013computing}.
		At the beginning, we insert in the queue all $c$-intervals, for $c\in\Sigma$. For each such interval $\langle L,R \rangle$ we set $LCP[R+1]=\ell = 0$. It is easy to see that after this step all and only the node-type LCP values equal to 0 are correctly set. 
		Assume, by induction, that all node-type LCP values less than or equal to $\ell$ have been correctly set, and we are about to extract from the queue the first triple $\langle L,R,\ell+1 \rangle$ having length $\ell+1$. For each extracted triple with length $\ell+1$ associated to a string $W$, consider the triple $\langle L',R',\ell+2 \rangle$ associated to one of its left-extensions $cW$. If $LCP[R'+1] \neq \bot$, then we have nothing to do. However, if $LCP[R'+1] = \bot$, then it must be the cases that (i) the value to write in this cell satisfies $LCP[R'+1] \geq \ell+1$, since by induction we have already filled all node-type LCP values smaller than or equal to $\ell$, and (ii) $LCP[R'+1]$ is of node-type, since otherwise the BWT interval of $cW$ would also include position $R'+1$. On the other hand, it cannot be the case that $LCP[R'+1] > \ell+1$ since otherwise the $cW$-interval would include position $R'+1$. We therefore conclude that $LCP[R'+1] = \ell+1$, so the algorithm correctly sets this node-type LCP value.

		The above argument settles correctness; to prove completeness, assume that, at some point, $LCP[i] = \bot$ and the correct value to be written in this cell is $\ell+1$. We want to show that we will pull a triple $\langle L,R,\ell+1 \rangle$ from the queue corresponding to a string $W$ (note that $\ell+1=|W|$ and, moreover, $W$ could end with $\#$) such that one of the left-extensions $aW$ of $W$ satisfies $\mathtt{range(aW)} = \langle L',i-1 \rangle$, for some $L'$. This will show that, at some point, we will set $LCP[i] \leftarrow \ell+1$. 
		We proceed by induction on $|W|$.
		Note that we separately set all LCP values equal to 0.
		The base case $|W|=1$ is easy: by the way we initialized the queue, $\langle \mathtt{range(c)}, 1\rangle$, for all $c\in\Sigma$, are the first triples we pop. Since we left-extend these ranges with all alphabet's characters except $\#$, it is easy to see that all LCP values equal to 1 are set. From now on we can therefore assume that we are setting LCP-values equal to $\ell+1>1$, i.e. $W=b\cdot V$, for $b\in\Sigma-\{\#\}$ and $V\in \Sigma^+$.
		Let $abV$ be 
		the length-$(\ell+2)$ left-extension of $W=bV$ 
		such that $\mathtt{right(abV)+1} = i$. Since, by our initial hypothesis, $\mathtt{LCP[i]} = \ell+1$, the collection contains also a suffix $aU$ lexicographically larger than $abV$ and such that $\mathtt{LCP(aU,abV)} = \ell+1$. But then, it must be the case that $\mathtt{LCP(right(bV)+1)} = \ell$ (it cannot be smaller by the existence of $U$ and it cannot be larger since $|bV|=\ell+1$). By inductive hypothesis, this value was set after popping a triple $\langle L'', R'', \ell\rangle$ corresponding to string $V$, left-extending $V$ with $b$, and pushing $\langle \mathtt{range(bV)}, \ell+1 \rangle$ in the queue. This completes the completeness proof since we showed that $\langle \mathtt{range(bV)}, \ell+1 \rangle$ is in the queue, so sooner or later we will pop it, extend it with $a$, and set $\mathtt{LCP[right(abV)+1] = LCP[i]} \leftarrow \ell+1$.

		If the queue uses too much space, then Algorithm \ref{alg:fill nodes} switches to a stack and Lines \ref{new stack2}-\ref{push4} are executed instead of Line \ref{beller et al.}. Note that this pseudocode fragment corresponds to Belazzougui's enumeration algorithm, except that now we also set LCP values in Line \ref{LCP in Node}. By the enumeration procedure's correctness, we have that, in Line \ref{LCP in Node}, $\langle \mathtt{first_W[1]}, \mathtt{first_W[t+1]} \rangle$ is the SA-range of a right-maximal string $W$ with $\ell = |W|$, and $\mathtt{first_W[i]}$ is the first position of the SA-range of $Wc_i$, with $i=1,\dots,t$, where $c_1, \dots, c_2$ are all the (sorted) right-extensions of $W$. Then, clearly each LCP value in Line \ref{LCP in Node} is of node-type and has value $\ell$, since it is the LCP between two strings prefixed by $W\cdot \mathtt{chars_W[i-1]}$ and $W\cdot \mathtt{chars_W[i]}$. Similarly,  completeness of the procedure follows from the completeness of the enumeration algorithm. Let $LCP[i]$ be of node-type. Consider the prefix $Wb$ of length $LCP[i]+1$ of the $i$-th suffix in the lexicographic ordering of all strings' suffixes. Since $LCP[i] = |W|$, the $(i-1)$-th suffix is of the form $Wa$, with $b\neq a$, and $W$ is right-maximal. But then, at some point our enumeration algorithm will visit the representation of $W$, with $|W|=\ell$. Since $i$ is the first position of the range of $Wb$, we have that $i= \mathtt{first_W[j]}$ for some $j \geq 2$, and Line \ref{LCP in Node} correctly sets $LCP[\mathtt{first_W[j]}] = LCP[i] \leftarrow \ell = |W|$.
		
		\textbf{Correctness and completeness - Algorithm \ref{alg:fill leaves}}. Proving correctness and completeness of this procedure is much easier. It is sufficient to note that the \texttt{while} loop iterates over all ranges $\langle L,R \rangle$ of strings ending with $\#$ and not containing $\#$ anywhere else (note that we start from the range of $\#$ and we proceed by recursively left-extending this range with symbols different than $\#$). Then, for each such range we set $LCP[L+1,R]$ to $\ell$, the string depth of the corresponding string (excluding the final character $\#$). It is easy to see that each leaf-type LCP value is correctly set in this way.
		
		\textbf{Complexity - Algorithm \ref{alg:fill nodes}}  If $\sigma > \sqrt n/\log^2 n$, then we run Beller et al's algorithm, which terminates in $O(n\log\sigma)$ time and uses $O(n) = o(n\log\sigma)$ bits of additional working space. Otherwise, we perform a linear number of operations on the stack since, as observed in Section \ref{sec:belazzougui}, the number of Weiner links is linear. By the same analysis of Section \ref{sec:belazzougui}, the operation in Line \ref{range distinct2} takes $O(k\log\sigma)$ amortized time on wavelet trees, and sorting in Line \ref{sort2} (using any comparison-sorting algorithm sorting $m$ integers in $O(m\log m)$ time) takes $O(k\log\sigma)$ time. Note that in this sorting step we can afford storing in temporary space nodes $x_1, \dots, x_k$ since this takes additional space $O(k\sigma\log n) = O(\sigma^2\log n) = O(n/\log^3n) = o(n)$ bits. 
		Overall, all these operations sum up to $O(n\log\sigma)$ time. Since the stack always takes at most $O(\sigma^2\log^2n)$ bits and $\sigma \leq \sqrt n/\log^2 n$, the stack's size never exceeds $O(n/\log^2n) = o(n)$ bits.
		
		\textbf{Complexity - Algorithm \ref{alg:fill leaves}} Note that, in the \texttt{while} loop, we start from the interval of $\#$ and recursively left-extend with characters different than $\#$ until this is possible. It follows that we visit the intervals of all strings of the form $W\#$ such that $\#$ does not appear inside $W$. Since these intervals form a cover of $[1,n]$, their number (and therefore the number of iterations in the \texttt{while} loop) is also bounded by $n$. This is also the maximum number of operations performed on the queue/stack. Using Beller et al.'s implementation for the queue and a simple vector for the stack, each operation takes constant amortized time. Operating on the stack/queue takes therefore overall $O(n)$ time.
		For each interval $\langle L,R \rangle$ popped from the queue/stack, in Line \ref{LCP in Leaves} we set $R-L-2$ LCP values. As observed above, these intervals form a cover of $[1,n]$ and therefore Line \ref{LCP in Leaves} is executed no more than $n$ times. Line \ref{getIntervals1} takes time $O(k\log\sigma)$. Finally, in Line \ref{sort1} we sort at most $\sigma$ intervals. Using any fast comparison-based sorting algorithm, this costs overall at most $O(n\log\sigma)$ time. 
		
		As far as the space usage of Algorithm \ref{alg:fill leaves} is concerned, note that we always push just pairs interval/length ($O(\log n)$ bits) in the queue/stack. If $\sigma > n/\log^3n$, we use Beller et al.'s queue, taking at most $O(n) = o(n\log\sigma)$ bits of space. Otherwise, the stack's size never exceeds $O(\sigma\cdot \log n)$ elements, with each element taking $O(\log n)$ bits. This amounts to $O(\sigma\cdot \log^2 n) = O(n/\log n) = o(n)$ bits of space usage. Moreover, in Lines \ref{getIntervals1}-\ref{sort1} it holds $\sigma\leq n/\log^3n$ so we can afford storing temporarily all intervals returned by $\mathtt{getIntervals}$ in $O(k\log n) = O(\sigma\log n) = O(n/\log^2n) = o(n)$ bits.
	\end{proof}
	
	\section{Merging BWTs}\label{app:merge}
	
	\begin{algorithm}[th!]
		\caption{$\mathtt{Merge(BWT_1,BWT_2, DA)}$}
		\label{alg:merge}
		
		\SetKwInOut{Input}{input}
		\SetKwInOut{Output}{behavior}
		\SetSideCommentLeft
		\LinesNumbered
		
		\Input{Wavelet trees of the Burrows-Wheeler transformed collections $\mathtt{BWT_1}\in [1,\sigma]^{n_1}$, $\mathtt{BWT_2}\in [1,\sigma]^{n_2}$ and empty document array $\mathtt{DA[1..n]}$, with $n=n_1+n_2$.}
		\Output{Computes Document Array $DA$.}
		\BlankLine	
		\BlankLine
		
		\eIf{$\sigma > n/\log^3n$}{
			
			$\mathtt P \leftarrow \mathtt{new\_queue()}$\tcc*[r]{Initialize new queue of interval pairs}\label{new queue3}
			
		}{
			
			$\mathtt P \leftarrow \mathtt{new\_stack()}$\tcc*[r]{Initialize new stack of interval pairs}\label{new stack3}
			
		}
		
		\BlankLine
		
		$\mathtt P\mathtt{.push(BWT_1.range(\#),BWT_2.range(\#))}$\tcc*[r]{Push SA-ranges of terminator}\label{push5}
		
		\BlankLine
		
		\While{$\mathtt{\mathbf{not}\ P.empty()}$}{\label{while3}
			\BlankLine
			
			$\langle L_1,R_1, L_2, R_2 \rangle \leftarrow \mathtt{P.pop()}$\tcc*[r]{Pop highest-priority element}\label{pop3}	
			
			\BlankLine
			
			\For{$i=L_1+L_2-1\dots L_2+R_1-1$}{
				$\mathtt{DA}[i] \leftarrow 0$\tcc*[r]{Suffixes from $\mathcal C_1$}\label{DA1}
			}	
			\For{$i=L_2+R_1\dots R_1+R_2$}{
				$\mathtt{DA}[i] \leftarrow 1$\tcc*[r]{Suffixes from $\mathcal C_2$}\label{DA2}
			}	
			
			\BlankLine

			\eIf{$\sigma > n/\log^3n$}{
				
				$\mathtt{P.push(getIntervals(L_1, R_1, L_2, R_2, BWT_1, BWT_2))}$\tcc*[r]{New intervals}\label{push8}

			}{
				
				$c_1^1, \dots, c_{k_1}^1 \leftarrow \mathtt{BWT_1.range\_distinct(L_1,R_1)}$\;\label{range distinct3}
				$c_1^2, \dots, c_{k_2}^2 \leftarrow \mathtt{BWT_2.range\_distinct(L_2,R_2)}$\;\label{range distinct4}
				$\{c_1\dots c_k\} \leftarrow \{c_1^1, \dots, c_{k_1}^1\} \cup \{c_1^2, \dots, c_{k_2}^2\}$\;\label{range distinct5}
				
				\BlankLine
				
				\For{$i=1\dots k$}{
					$\langle L_1^i, R_1^i\rangle \leftarrow \mathtt{BWT_1.bwsearch(\langle L_1, R_1\rangle, c_i)}$\tcc*[r]{Backward search step}\label{BWS2}
				}
				
				\For{$i=1\dots k$}{
					$\langle L_2^i, R_2^i\rangle \leftarrow \mathtt{BWT_2.bwsearch(\langle L_2, R_2\rangle, c_i)}$\tcc*[r]{Backward search step}\label{BWS3}
				}

				\BlankLine
				
				$\langle \hat L_1^i, \hat R_1^i, \hat L_2^i, \hat R_2^i, \rangle_{i=1, \dots, k} \leftarrow \mathtt{sort}(\langle L_1^i, R_1^i, L_2^i, R_2^i, \rangle_{i=1, \dots, k})$\;\label{sort3}
				
				\BlankLine
				
				\For{$i=k\dots 1$}{
					
					$\mathtt{P.push}(\hat L_1^i, \hat R_1^i, \hat L_2^i, \hat R_2^i)$\tcc*[r]{Push in order of decreasing length}\label{push6}
					
				}
		
			}

		}
	\end{algorithm}
	
	Note that Algorithm \ref{alg:merge} is similar to Algorithm \ref{alg:fill leaves}, except that now we manipulate pairs of intervals. In Line \ref{sort3}, we sort quadruples according to the length $R_1^i + R_2^i - (L_1^i + L_2^i) +2$ of the combined interval on $BWT(\mathcal C_1\cup \mathcal C_2)$. Finally, note that Backward search can be performed correctly also when the input interval is empty: $\mathtt{BWT_j.bwsearch(\langle L_j, L_j-1 \rangle, c)}$, where $L_j-1$ is the number of suffixes in $\mathcal C_j$ smaller than some string $W$, correctly returns the pair $\langle L', R'\rangle$ such that $L'$ is the number of suffixes in $\mathcal C_j$ smaller than $cW$: this is true when implementing backward search with a $rank_c$ operation on position $L_j$; then, if the original interval is empty we just set $R'=L'-1$ to keep the invariant that $R'-L'+1$ is the interval's length. %Correctness, completeness, and complexity analysis of the algorithm are completely analogous to those of Algorithm \ref{alg:fill leaves}.

	\section{Enumerating Suffix Tree Intervals in Succinct Space}\label{sec:ST}
	
	We note that Algorithm \ref{alg:fill nodes} can be used to enumerate suffix tree intervals using just $o(n\log\sigma)$ space on top of the input BWT of a single text, when this is represented with a wavelet tree. This is true by definition in Belazzougui's procedure (Lines \ref{new stack2}-\ref{push4}), but a closer look reveals that also Beller et al's procedure (Line \ref{beller et al.}) enumerates suffix tree intervals. At each step, we pop from the queue an element $\langle\langle L,R \rangle, |W| \rangle$ with $\langle L,R \rangle = \mathtt{range(W)}$ for some string $W$, left-extend the range with all $a\in \mathtt{BWT.rangeDistinct(L,R)}$, obtaining the ranges $\mathtt{range(aW)} = \langle L_a, R_a\rangle$ and, only if $LCP[R_a+1]=\bot$, set $LCP[R_a+1] \leftarrow |W|$ and push $\langle \langle L_a, R_a\rangle, |W|+1 \rangle$ on the stack. But then, since $LCP[R_a+1] = |W|$ we have that the $R_a$-th and $(R_a+1)$-th  smallest suffixes start, respectively, with $aUc$ and $aUd$ for some $c<d\in\Sigma$, where $W=Uc$. This implies that $aU$ is right-maximal, and the corresponding suffix tree node has at least two children labeled $c$ and $d$; in particular, $\langle L_a, R_a\rangle$ is the range of $aW = aUc$, that is, one of these two children. Since we assume that we are working with a single text, $\langle L_a, R_a\rangle$ is the range of a suffix tree node if and only if $R_a>L_a$: in this case, we return this range. 
	Completeness of the visit (i.e. we return all suffix tree nodes' intervals) follows from the completeness of the LCP-array construction procedure (i.e. we fill all LCP values).
	To conclude note that, to perform our visit, we are using the array $LCP$ to store null/non-null entries (i.e. $\bot$ or any other number); this shows that we do not actually need the LCP array: a bitvector of length $n = o(n\log\sigma)$ is sufficient (remember that Beller et al.'s strategy is used on large alphabets, so $n = o(n\log\sigma)$ holds).
	
\end{document}